\newcommand{\aiversion}[1]{}
\newcommand{\longversion}[1]{#1}

\longversion{
\documentclass[usletter]{article}
\usepackage[totalwidth=450pt,totalheight=640pt]{geometry}
\usepackage{fixbib}

}

\aiversion{
\documentclass[letterpaper]{article}
\usepackage{aaai}
\usepackage{pgf}
\usepackage{aaai}
\usepackage{times}
\usepackage{helvet}
\usepackage{courier}
\setlength{\pdfpagewidth}{8.5in}
\setlength{\pdfpageheight}{11in}

\frenchspacing
\pdfinfo{
  /Author (Iyad Kanj, Stefan Szeider)
  /Title (On the Subexponential Time Complexity of CSP)
  \setcounter{secnumdepth}{0}
}
}

\usepackage{times}

\newcommand{\citex}[1]{\citeauthor{#1}~\shortcite{#1}}

\usepackage{microtype,verbatim,booktabs}
\usepackage{graphicx}

\usepackage{amsthm,amsmath,amssymb,bbm}

\newtheorem{lemma}{Lemma}
\newtheorem{theorem}{Theorem}
\newtheorem{proposition}{Proposition}

\newcommand{\Card}[1]{|#1|}

\let\phi=\varphi
\let\epsilon=\varepsilon

\newcommand{\CCC}{\mathcal{C}}

\newcommand{\var}{\mathit{var}}

\newcommand{\ARITY}{\mbox{\rm \textsf{arity}}}

\newcommand{\VARS}{\mbox{\rm \textsf{vars}}}
\newcommand{\DOM}{\mbox{\rm \textsf{dom}}}

\newcommand{\CONS}{\mbox{\rm \textsf{cons}}}
\newcommand{\DEG}{\mbox{\rm \textsf{deg}}}
\newcommand{\SIZE}{\mbox{\rm \textsf{size}}}
\newcommand{\TUPLES}{\mbox{\rm \textsf{tuples}}}

\newcommand{\TW}{\mbox{\rm \textsf{tw}}}

\newcommand{\Nat}{\mathbb{N}}

\newcommand{\nop}[1]{}

\newcommand{\msc}[1]{\textsc{#1}}

\title{On the Subexponential Time Complexity of CSP} \longversion{
  \author{Iyad Kanj\\
    \small School of Computing,
    DePaul University\\[-4pt]
    \small     Chicago, USA\\[-4pt]
    \small \texttt{ikanj@cs.depaul.edu} \and Stefan
    Szeider\thanks{Research supported by the European Research Council
      (ERC), project COMPLEX REASON 239962.
    }\\
    \small Vienna University of Technology\\[-4pt]
    \small Vienna, Austria\\[-4pt]
    \small \texttt{stefan@szeider.net} } \date{} } \aiversion{
  \author{ \textbf{\Large Iyad Kanj}$^1$ \and
    \textbf{\Large Stefan Szeider}$^2$\\
    \mbox{}$^1$School of Computing, DePaul University, Chicago, USA\\
    ikanj@cs.depaul.edu \\
    \mbox{}$^2$Institute of Information Systems, Vienna University of
    Technology,
    Vienna, Austria\\
    stefan@szeider.net }
}

\begin{document}

\maketitle

\begin{abstract}

\aiversion{\begin{quote}}
\noindent  A \msc{CSP} with $n$ variables ranging over a domain of $d$ values
  can be solved by brute-force in $d^n$ steps (omitting a polynomial
  factor). With a more careful approach, this trivial upper bound can
  be improved for certain natural restrictions of the \msc{CSP}.
In
  this paper we establish theoretical limits to such improvements, and
  draw a detailed landscape of the subexponential-time complexity of
  \msc{CSP}.

  We first establish relations between the subexponential-time
  complexity of \msc{CSP} and that of other problems, including
  \msc{CNF-Sat}. We exploit this connection to provide tight
  characterizations of the subexponential-time complexity of \msc{CSP}
  under common assumptions in complexity theory. For several natural
  \msc{CSP} parameters, we obtain threshold functions
  that precisely dictate the subexponential-time complexity of
  \msc{CSP} with respect to the parameters under consideration.

  Our analysis provides fundamental results indicating {\em whether and when} one
  can significantly improve on the brute-force search approach for solving
  \msc{CSP}.
\aiversion{\end{quote}}
 \end{abstract}

\longversion{\thispagestyle{empty}}
\pagestyle{plain}
\section{Introduction}\label{sec:intro}
The Constraint Satisfaction Problems (CSP) provides a general and
uniform framework for the representation and solution of hard
combinatorial problems that arise in various areas of Artificial
Intelligence and Computer Science \cite{RossiVanBeekWalsh06}. For instance, in database theory, the CSP is equivalent to the evaluation
problem of conjunctive queries on relational databases
\cite{GottlobLeoneScarcello02b}.

It is well known that CSP is NP-hard, as it entails fundamental NP-hard
problems such as \msc{3-Colorability} and \msc{3-CNF-Sat}. Hence, we cannot hope for a \emph{polynomial-time algorithm} for
CSP\@. On the other hand, CSP can obviously be solved in
\emph{exponential time}: by simply trying all possible instantiations
of the variables, we can solve a CSP instance consisting of $n$
variables that range over a domain of $d$ values in time $d^n$
(omitting a polynomial factor in the input size). Significant work has been concerned
with improving this trivial upper bound
\cite{FederMotwani02,BeigelEppstein05,GrandoniItaliano06}, in particular, for certain restrictions of CSP.
For instance, binary CSP with domain size $d$ can now be solved in time
$(d-1)^n$ (omitting a polynomial factor in the input size) by a
forward-checking algorithm employing a fail-first variable ordering
heuristic~\cite{Razgon05}. All these improvements over the trivial brute-force search give
exponential running times in which the exponent is linear in~$n$.

The aim of this paper is to investigate the theoretical limits of such
improvements.  More precisely, we explore whether the exponential
factor $d^n$ can be reduced to a \emph{subexponential factor}
$d^{o(n)}$ or not, considering various natural NP-hard restrictions of
the~CSP. We note that the study of the existence of
subexponential-time algorithms is of prime interest, as a
subexponential-time algorithm for a problem would allow us to solve
larger hard instances of the problem in comparison to an
exponential-time algorithm.

\paragraph{Results}

We obtain lower and upper bounds and draw a detailed complexity
landscape of \msc{CSP} with respect to subexponential-time
solvability.  Our lower bounds are subject to (variants of) the
\emph{Exponential Time Hypothesis} (ETH), proposed by
\citex{ImpagliazzoPaturi01}, which states that \msc{3-CNF-Sat} has no
subexponential-time algorithm.

It is easy to see that \msc{CSP} of \emph{bounded domain size} (i.e., the
maximum number of values for each variable) and \emph{bounded arity} (i.e.,
the maximum number of variables that appear together in a constraint)
has a subexponential-time algorithm if and only if the ETH fails. Our
first result provides evidence that when we drop the bound on the domain
size or the bound on the arity, the problem becomes ``harder'' (we
refer to the discussion preceding Proposition~\ref{thm:clique}):

\begin{enumerate}
\item If \msc{Boolean CSP} is solvable in nonuniform subexponential time
  then so is (unrestricted) \msc{CNF-Sat}.
\item If \msc{2-CSP} (all constraints have arity 2) is solvable in
  subexponential time then \msc{Clique} is solvable in time $N^{o(k)}$ ($N$ is the number of vertices and $k$ is the clique-size).
\end{enumerate}
As it turns out, the number of tuples plays an important role in characterizing the subexponential time complexity of \msc{CSP}. We show the following tight result:

\begin{enumerate}
\setcounter{enumi}{2}
\item \msc{CSP} is solvable in subexponential time for instances in which the
number of tuples is $o(n)$, and unless the ETH fails, is not solvable in
subexponential time if the number of tuples in the instances is
$\Omega(n)$.
\end{enumerate}
For Boolean \msc{CSP} of linear size we can even derive an equivalence to the ETH:
\begin{enumerate}
\setcounter{enumi}{3}
\item Boolean \msc{CSP} for instances of size $\Omega(n)$ is solvable in
  subexponential time if and only if the ETH fails.
\end{enumerate}
Results~3\ and 4\ also hold if we consider the total number of tuples in
the constraint relations instead of the input size.

By a classical result of \cite{Freuder90}, \msc{CSP} becomes easier if the
instance has \emph{small treewidth}. There are several ways of measuring
the treewidth of a \msc{CSP} instance, depending on the graph used to
model the structure of the instance. The most common models are the
primal graph and the incidence graph. The former has as vertices the
variables of the \msc{CSP} instance, and two variables are adjacent if
they appear together in a constraint. The incidence graph is the
bipartite graph on the variables and constraints, where a variable is
incident to all the constraints in which it is involved.  We show that
the treewidth of these two graph models give rise to different
subexponential-time complexities:
\begin{enumerate}
\setcounter{enumi}{4}
\item \msc{CSP} is solvable in subexponential time for instances whose primal
  treewidth is $o(n)$, but is not solvable in subexponential time for
  instances whose primal treewidth is $\Omega(n)$, assuming the ETH.
\item \msc{CSP} is solvable in polynomial time for instances whose incidence
  treewidth is $O(1)$, but is not solvable in subexponential time for
  instances whose incidence treewidth is $\omega(1)$ unless the ETH fails.
\end{enumerate}
Our tight results, summarized in the table at the end of this paper,
provide strong theoretical evidence that some of the natural
restrictions of \msc{CSP} may be ``harder than''
\msc{$k$-CNF-Sat}---for which a subexponential-time algorithm would
lead to the failure of the ETH. Hence, our results provide a new point
of view of the relationship between SAT and \msc{CSP}, an important
topic of recent AI research
\cite{JeavonsPetke12,DimopoulosSterigou06,BenhamouParisSigel12,Bennaceur04}.
\aiversion{
 Some
proofs and details are omitted from the current paper due to the lack
of space. A complete version of the paper is available
as~\cite{kanjszeiderarxiv}.
}

\section{Preliminaries}
\label{sec:prelim}

\longversion{\subsection{Constraint satisfiability and CNF-satisfiability}
\label{subsec:prelcsp}}
\aiversion{\paragraph{Constraint satisfiability and CNF-satisfiability}}

An instance $I$ of the \msc{Constraint Satisfaction Problem} (or
\msc{CSP}, for short) is a triple~$(V, D, \mathcal{C})$, where $V$ is a
finite set of \emph{variables}, $D$ is a finite set of \emph{domain
  values}, and $\mathcal{C}$ is a finite set of \emph{constraints}.
Each constraint in~$\mathcal{C}$ is a pair $(S,R)$, where~$S$, the
\emph{constraint scope}, is a non-empty sequence of distinct variables
of~$V$, and $R$, the \emph{constraint relation}, is a relation over~$D$
whose arity matches the length of~$S$; a relation is considered as a set
of tuples.  Therefore, the \emph{size} of a \msc{CSP} instance
$I=(V,D,\CCC)$ is the sum $\sum_{(S, R) \in \mathcal{C}} |S| \cdot |R|$;
the \emph{total number of tuples} is $\sum_{(S, R) \in
  \mathcal{C}} |R|$.
We assume, without loss of generality, that every variable
occurs in at least one constraint scope and every domain element occurs
in at least one constraint relation.  Consequently, the size of an instance $I$ is at least as large as the number of variables in $I$. We write $\var(C)$ for the set of
variables that occur in the scope of
constraint~$C$. 

An \emph{assignment} or \emph{instantiation} is a mapping from the set
$V$ of variables to the domain~$D$. An assignment $\tau$
\emph{satisfies} a constraint $C=((x_1,\dots,x_n),R)$ if
$(\tau(x_1),\dots,\tau(x_n))\in R$, and $\tau$ satisfies the \msc{CSP} instance
if it satisfies all its constraints. An instance~$I$ is
\emph{consistent} or \emph{satisfiable} if it is satisfied by some
assignment.   \msc{CSP}  is the problem of deciding whether a
given instance of \msc{CSP} is consistent. \msc{Boolean CSP}
denotes the \msc{CSP} with the \emph{Boolean domain} $\{0,1\}$. By \msc{$r$-CSP} we denote the restriction of \msc{CSP} to instances in which the arity of each constraint is at most $r$.


For an instance $I=(V,D,\CCC)$ of \msc{CSP} we define the following
basic parameters:
\begin{itemize}
\item $\VARS$: the number $\Card{V}$ of variables, usually denoted by $n$;
\item $\SIZE$: the size of the CSP instance;
\item $\DOM$: the number $\Card{D}$ of values;
\item $\CONS$: the number $\Card{\CCC}$ of constraints;

\end{itemize}
\msc{CNF-Sat} is the satisfiability problem for propositional formulas
in conjunctive normal form (CNF).  \msc{$k$-CNF-Sat} denotes
\msc{CNF-Sat} restricted to formulas where each clause is of width at
most $k$, i.e., contains at most $k$ literals.

\longversion{\subsection{Subexponential time}
\label{subsec:prelsubexp}}
\aiversion{\paragraph{Subexponential time}}

The time complexity functions used in this paper are assumed to be
proper complexity functions that are unbounded and nondecreasing.  The
$o(\cdot)$ notation used denotes the $o^{\mbox{eff}}(\cdot)$ notation
\cite{FlumGrohe06}.  More formally, for any two computable functions $f,
g: \mathbb{N} \rightarrow \mathbb{N}$, by writing $f(n) = o(g(n))$ we
mean that there exists a computable nondecreasing unbounded function
$\mu(n) : \mathbb{N} \rightarrow \mathbb{N}$, and $n_0 \in \mathbb{N}$,
such that $f(n) \leq g(n)/\mu(n)$ for all $n \geq n_0$.

It is clear that \msc{CSP} and \msc{CNF-Sat} are solvable in time
$\DOM^{n} |I|^{O(1)}$ and $2^n |I|^{O(1)}$, respectively, where $I$ is
the input instance and $n$ is the number of variables in $I$.  We say
that the \msc{CSP} (resp. \msc{CNF-Sat}) problem is solvable in {\em
  uniform subexponential time} if there exists an algorithm that
solves the problem in time $\DOM^{o(n)}|I|^{O(1)}$
(resp. $2^{o(n)}|I|^{O(1)}$). Using the results
of~\cite{ChenKanjXia09,FlumGrohe06}, the above definition is
equivalent to the following: The \msc{CSP} (resp. \msc{CNF-Sat})
problem is solvable in {\em uniform subexponential time} if there
exists an algorithm that for all $\epsilon = 1/\ell$, where $\ell$ is
a positive integer, solves the problem in time $\DOM^{\epsilon
  n}|I|^{O(1)}$ (resp. $2^{\epsilon n}|I|^{O(1)}$).  The \msc{CSP}
(resp. \msc{CNF-Sat}) problem is solvable in {\em nonuniform
  subexponential time} if for each $\epsilon = 1/\ell$, where $\ell$
is a positive integer, there exists an algorithm $A_{\epsilon}$ that
solves the problem in time $\DOM^{\epsilon n} |I|^{O(n)}$
(resp. $2^{\epsilon n} |I|^{O(1)}$) (that is, the algorithm depends on
$\epsilon$).  We note that subexponential-time algorithms running in
$O(2^{\sqrt{n}})$ time do exist for many natural
problems~\cite{AlberFernauNiedermeier04}.

Let $Q$ and $Q'$ be two problems, and let $\mu$ and $\mu'$ be two
parameter functions defined on instances of $Q$ and $Q'$,
respectively. In the case of \msc{CSP} and \msc{CNF-Sat}, $\mu$ and
$\mu'$ will be the number of variables in the instances of these
problems. A {\em subexponential-time Turing reduction
  family}
\longversion{\cite{ImpagliazzoPaturiZane01} (see also~\cite{FlumGrohe06}),
  shortly a serf-reduction\footnote{Serf-reductions were
      introduced by Impagliazzo et
      al.~\cite{ImpagliazzoPaturiZane01}. Here we use the definition given
      by Flum and Grohe~\cite{FlumGrohe06}. There is a slight difference
      between the two definitions, and the latter definition is more
      flexible for our purposes.},
  }
\aiversion{\cite{ImpagliazzoPaturiZane01,FlumGrohe06},
  shortly a \emph{serf-reduction},
}
is an algorithm $A$ with an oracle to $Q'$ such that there are
computable functions $f, g: \Nat \longrightarrow \Nat$
satisfying: (1) given a pair $(I, \epsilon)$ where $I \in Q$ and
$\epsilon = 1/\ell$ ($\ell$ is a positive integer), $A$ decides $I$ in
time $f(1/\epsilon)\DOM^{\epsilon \mu (I)} |I|^{O(1)}$ (for
\msc{CNF-Sat} $\DOM =2$); and (2) for all oracle queries of the form
``$I' \in Q'$'' posed by $A$ on input $(I, \epsilon)$, we have $\mu'(I')
\leq g(1/\epsilon)(\mu(I) + \log{|I|})$.

The optimization class SNP consists of all search problems expressible
by second-order existential formulas whose first-order part is universal
\cite{PapadimitriouYannakakis91}. \citex{ImpagliazzoPaturiZane01}
introduced the notion of {\em completeness} for the class SNP under
serf-reductions, and identified a class of problems which are complete
for SNP under serf-reductions, such that the subexponential-time
solvability for any of these problems implies the subexponential-time
solvability of all problems in SNP\@.  Many well-known NP-hard problems
are proved to be complete for SNP under the serf-reduction, including
\msc{$3$-Sat}, \msc{Vertex Cover}, and \msc{Independent Set}, for which
extensive efforts have been made in the last three decades to develop
subexponential-time algorithms with no success. This fact has led to the
{\it exponential-time hypothesis}, ETH, which is equivalent to the
statement that not all SNP problems are solvable in subexponential-time:

\smallskip
\noindent\fbox{\begin{minipage}{0.97\linewidth}
{\it Exponential-Time Hypothesis} (ETH): \ \
The problem \msc{$k$-CNF-Sat}, for any $k \geq 3$, cannot be solved in time $2^{o(n)}$,
where $n$ is the number of variables in the input formula. Therefore, there exists $c > 0$ such that \msc{$k$-CNF-Sat} cannot be solved in time $2^{cn}$.
\end{minipage}}
\smallskip

The following result is implied from
\cite[Corollary~1]{ImpagliazzoPaturiZane01} and from the proof of the
Sparsification Lemma~\cite{ImpagliazzoPaturiZane01}, \cite[Lemma
16.17]{FlumGrohe06}.

\begin{lemma}
\label{lem:subexpnm}
\msc{$k$-CNF-Sat} ($k \geq 3$) is solvable in $2^{o(n)}$ time if and only if \msc{$k$-CNF-Sat} with a linear number of clauses and in which the number of occurrences of each variable is upper bounded by a constant is solvable in time $2^{o(n)}$, where $n$ is the number of variables in the formula (note that the size of an instance of \msc{$k$-CNF-Sat} is polynomial in~$n$).
\end{lemma}
The ETH has become a standard hypothesis in complexity theory~\cite{LokshtanovMarxSuarabh11}.

We close this section by mentioning some further work on the
subexponential-time complexity of \msc{CSP}.
There are several results on \msc{2-CSP} with bounds on $\TW$, the
treewidth of the primal graph
\longversion{(see Section  \ref{sec:constraints} for definitions).}
\aiversion{(see the Introduction section for definitions).}
\citex{LokshtanovMarxSuarabh11} showed the following lower bound,
using a result on list coloring \cite{FellowsEtAl11}: \msc{2-CSP}
cannot be solved in time $f(\TW) n^{o(\TW)}$ unless the ETH fails.
\citex{Marx10b} showed that if there is a recursively enumerable class
${\cal G}$ of graphs with unbounded treewidth and a function $f$ such
that $\msc{2-CSP}$ can be solved in time $f(G)n^{o(\TW/\log \TW)}$ for
instances whose primal graph is in ${\cal G}$, then the ETH fails.
\citex{Traxler08} studied the subexponential-time complexity of
CSP where the constraints are represented by listing the forbidden
tuples (in contrast to the standard representation that we use, where
the allowed tuples are given, and which naturally captures database
problems
\cite{GottlobLeoneScarcello02b,Grohe06,PapadimitriouYannakakis99}). This
setting can be considered as a generalisation of \msc{CNF-Sat}; a single
clause gives rise to a constraint with exactly one forbidden tuple.

\section{Relations between \msc{CSP} and \msc{CNF-Sat}}
\label{sec:sattocsp}
In this section, we investigate the relation between the subexponential-time complexity of \msc{CSP} and that of \msc{CNF-Sat}.
A clause of constant width can be represented by a constraint of
constant arity; the reverse holds as well (we get a constant number of clauses). Hence, we have:
\begin{proposition}
\label{thm:boundedarity}
\msc{Boolean $r$-CSP} is solvable in subexponential time if and only if the ETH fails.
\end{proposition}

The following proposition suggests that
Proposition~\ref{thm:boundedarity} may not extend to \msc{$r$-CSP} with
unbounded domain size. Chen et al.~\cite{ChenEtal05} showed that if \msc{Clique}
(decide whether a given a graph on $N$ vertices contains a complete
subgraph of $k$ vertices) is solvable in time $N^{o(k)}$ then the ETH
fails. The converse, however, is generally believed not to be true.
The idea behind the proof of the proposition goes back to the paper
by \citex{PapadimitriouYannakakis99}, where
they used it in the context of studying the complexity of database
queries. We skip the proof, and refer the reader
to the original source~\cite{PapadimitriouYannakakis99}.

\begin{proposition}
\label{thm:clique}
If \msc{2-CSP} is solvable in subexponential time then \msc{Clique} is
solvable in time $N^{o(k)}$.
\end{proposition}

We explore next the relation between \msc{Boolean CSP} with unbounded
arity and \msc{CNF-Sat}. We show that if \msc{Boolean CSP} is solvable
in nonuniform subexponential time then so is \msc{CNF-Sat}. To do so, we
exhibit a nonuniform subexponential-time Turing reduction from \msc{CNF-Sat} to
\msc{Boolean CSP}.

Intuitively, one would try to reduce an instance $F$
of \msc{CNF-Sat} to an instance $I$ of \msc{CSP} by associating with
every clause in $F$ a constraint in $I$ whose variables are the
variables in the clause, and whose relation consists of all tuples that
satisfy the clause. There is a slight complication in such an
attempted reduction because the number of tuples in a constraint could
be exponential if the number of variables in the corresponding clause is
linear (in the total number of variables). To overcome this subtlety,
the idea is to first apply a subexponential-time (Turing) reduction,
which is originally due to \citex{Schuler05} and was also used and
analyzed by \citex{CalabroImpagliazzoPaturi06}, that reduces the
instance $F$ to subexponentially-many (in $n$) instances in which the
width of each clause is at most some constant $k$; in our case, however,
we will reduce the width to a suitable nonconstant value. We follow
this reduction with the reduction to \msc{Boolean CSP} described
above.

\begin{theorem}\label{thm:sattocsp}
If \msc{Boolean CSP} has a nonuniform subexponential-time algorithm then so does \msc{CNF-Sat}.
\end{theorem}

\begin{proof}
Suppose that \msc{Boolean CSP} is solvable in nonuniform subexponential time. Then for every $\delta > 0$, there exists an algorithm $A'_{\delta}$ that, given an instance $I$ of \msc{Boolean CSP} with $n'$ variables, $A'_{\delta}$ solves $I$ in time $2^{\delta n'}|I|^{c'}$, for some constant $c' > 0$.

Let $0 < \epsilon < 1$ be given. 
We describe an algorithm $A_{\epsilon}$ that solves \msc{CNF-Sat} in time $2^{\epsilon n} m^{O(1)}$. 
Set $k = \lfloor \frac{\epsilon n}{2(1+c')} \rfloor$. Let $F$ be an instance of \msc{CNF-Sat} with $n$ variables and $m$ clauses.
The algorithm $A_{\epsilon}$ is a search-tree algorithm, and works as follows. The algorithm picks a clause $C$ in $F$ of width more than $k$; if no such clause exists the algorithm stops. Let $l_1, \ldots, l_k$ be any $k$ literals in $C$. The algorithm branches on $C$ into two branches. The first branch, referred to as a {\em left branch}, corresponds to one of these $k$ literals being assigned the value 1 in the satisfying assignment sought, and in this case $C$ is replaced in $F$ by the clause $(l_1 \vee \ldots \vee l_k)$, thus reducing the number of clauses in $F$ of width more than $k$ by~1. The second branch, referred to as a {\em right branch}, corresponds to assigning all those $k$ literals the value 0 in the satisfying assignment sought; in this case the values of the variables corresponding to those literals have been determined, and the variables can be removed from $F$ and $F$ gets updated accordingly. Therefore, in a right branch the number of variables in $F$ is reduced by $k$. The execution of the part of the algorithm described so far can be depicted by a binary search tree whose leaves correspond to instances resulting from $F$ at the end of the branching, and in which each clause has width at most $k$. The running time of this part of the algorithm is proportional to the number of leaves in the search tree, or equivalently, the number of root-leaf paths in the search tree. Let $F'$ be an instance resulting from $F$ at a leaf of the search tree. We reduce $F'$ to an instance $I_{F'}$ of \msc{Boolean CSP} as follows. For each clause $C'$ in $F'$, we correspond to it a constraint whose variable-set is the set of variables in $C'$, and whose tuples consist of at most $2^k-1$ tuples corresponding to all assignments to the variables in $C'$ that satisfy $C'$. Clearly, $I_{F'}$ can be constructed in time $2^k m^{O(1)}$ (note that the number of clauses in $F'$ is at most $m$). To the instance $I_{F'}$, we apply the algorithm $A'_{\delta}$ with $\delta = \epsilon/2$. The algorithm $A_{\epsilon}$ accepts $F$ if and only if $A'_{\delta}$ accepts one of the instances $I_{F'}$, for some $F'$ resulting from $F$ at a leaf of the search tree.

The running time of $A_{\epsilon}$ is upper bounded by the number of
leaves in the search tree, multiplied by a polynomial in the length of
$F$ (polynomial in $m$) corresponding to the (maximum) total running
time along a root-leaf path in the search tree, multiplied by the time
to construct the instance $I_{F'}$ corresponding to $F'$ at a leaf of
the tree, and multiplied by the running time of the algorithm
$A'_{\delta}$ applied to $I_{F'}$. Note that the binary search tree depicting
the execution of the algorithm is not a complete binary tree. To upper bound the size of the search tree, let $P$ be a root-leaf path in the
search tree, and let $\ell$ be the number of right branches along
$P$. Since each right branch removes $k$ variables, $\ell \leq n/k$ and
the number of variables left in the instance $F'$ at the leaf endpoint
of $P$ is $n- \ell k$. Noting that the length of a path with $\ell$ right branches
is at most $m + \ell$ (each left branch reduces $m$ by 1 and hence there
can be at most $m$ such branches on $P$, and there are $\ell$ right
branches), we conclude that the number of root-leaf paths, and hence the number of leaves, in the search
tree is at most $\sum_{\ell=0}^{\lceil n/k \rceil} {m+\ell \choose \ell}$.

The reduction from $F'$ to an instance of
\msc{Boolean CSP} can be carried out in time $2^k m^{O(1)}$, and results
in an instance $I_{F'}$ in which the number of variables is at most
$n'=n- \ell k$, the number of constraints is at most $m$, and the total
size is at most $2^k m^{O(1)}$. Summing over all possible paths in the search
tree, the running time of $A_{\epsilon}$ is $2^{\epsilon n}
m^{O(1)}$.
\longversion{
This is a consequence of the following estimation:
\begin{eqnarray}
   \sum_{\ell=0}^{\lceil n/k \rceil} {m+\ell \choose \ell} 2^k m^{O(1)} \cdot 2^{\delta(n-\ell k)}. (2^k m^{O(1)})^{c'}   &\leq & 2^{(1+c')k+\delta n} m^{O(1)} \sum_{\ell=0}^{\lceil n/k\rceil } {m+\lceil n/k \rceil \choose \ell} \nonumber \\
    & \leq  & 2^{(1+c')k+\delta n} m^{O(1)} {2m \choose \lceil n/k\rceil} \label{ineq2} \\
    & \leq & 2^{(1+c')k+\delta n} m^{O(1)} \cdot (2m)^{n/k} \label{ineq3} \\
    & \leq & 2^{(1+c')k+\delta n} m^{O(1)} \label{ineq5} \\
    & \leq &  2^{\epsilon n} m^{O(1)}. \nonumber
\end{eqnarray}
The first inequality follows after replacing $\ell$ by the larger value $\lceil n/k \rceil$ in the upper part of the binomial coefficient, and upper bounding the term $2^{-\ell \delta k}$ by 1.
Inequality (\ref{ineq2}) follows from the fact that the largest binomial coefficient in the summation is ${m + \lceil n/k \rceil \choose \lceil n/k \rceil} \leq {2m \choose \lceil n/k \rceil}$ ($m \geq  \lceil n/k \rceil$, otherwise $m$ is a constant, and the instance of \msc{CNF-Sat} can be solved in polynomial time from the beginning), and hence, the summation can be replaced by the largest binomial coefficient multiplied by the number of terms ($\lceil n/k \rceil$+1) in the summation, which gets absorbed by the term $m^{O(1)}$. Inequality (\ref{ineq3}) follows from the trivial upper bound on the binomial coefficient (the ceiling can be removed because polynomials in $m$ get absorbed).
Inequality (\ref{ineq5}) follows after noting that $n/k$ is a constant (depends on $\epsilon$), and after substituting $k$ and $\delta$ by their values/bounds.
}

It follows that the algorithm $A_{\epsilon}$ solves \msc{CNF-Sat} in time $2^{\epsilon n} m^{O(1)}$. Therefore, if \msc{Boolean CSP} has a nonuniform subexponential-time algorithm, then so does \msc{CNF-Sat}. The algorithm is nonuniform because the polynomial factor in the running time (exponent of $m$) depends on $\epsilon$.
\end{proof}

\section{Instance size and number of tuples}
\label{sec:size}
In this section we give characterizations of the subexponential-time complexity of \msc{CSP} with respect to the instance size and the number of tuples. Recall that the size of an instance $I =(V, D, \mathcal{C})$ of \msc{CSP} is $\SIZE = \sum_{(S, R) \in \mathcal{C}} |S| \cdot |R|$. We also show that the subexponential-time solvability of \msc{Boolean CSP} with linear size, or linear number of tuples, is equivalent to the statement that the ETH fails.

\begin{lemma}
\label{lem:hardnesssize}
Unless the ETH fails,  \msc{Boolean CSP} is not solvable in
subexponential-time if the instance size is $\Omega(n)$.
\end{lemma}

\begin{proof}
Let $s(n) = \Omega(n) \geq cn$ be a complexity function, where $c > 0$ is a constant. Suppose that the restriction of \msc{CSP} to instances of size at most $s(n)$ is solvable in subexponential time, and we will show that \msc{3-CNF-Sat} is solvable in subexponential time. By Lemma~\ref{lem:subexpnm}, it is sufficient to show that \msc{3-CNF-Sat} with a linear number of clauses is solvable in $2^{o(n)}$ time. Using a padding argument, we can prove the preceding statement assuming any linear upper bound on the number of clauses; we pick this linear upper bound to be $cn/24$, where $c$ is the constant in the upper bound on $s(n)$.

Let $F$ be an instance of \msc{3-CNF-Sat} with $n$ variables and at most $cn/24$ clauses. We reduce $F$ to an instance $I_F$ of \msc{Boolean CSP} using the same reduction described in the proof of Theorem~\ref{thm:sattocsp}: for each clause $C$ of $F$ we correspond a constraint whose variables are those in $C$ and whose tuples are those corresponding to the satisfying assignments to $C$. Since the width of $C$ is 3 and the number of clauses is at most $cn/24$, the instance $I_{F}$ consists of at most $cn/24$ constraints, each containing at most 3 variables and 8 tuples. Therefore, the size of $I_{F}$ is at most $cn$. We now apply the hypothetical subexponential-time algorithm to $I_{F}$. Since $|I|$ is linear in $n$, and since the reduction takes linear time in $n$, we conclude that \msc{3-CNF-Sat} is solvable in time $2^{o(n)} n^{O(1)} = 2^{o(n)}$. The proof follows.
\end{proof}

\begin{lemma}\label{lem:tupleseasy}
\msc{CSP} restricted to instances with $o(n)$ tuples is solvable in subexponential-time.
\end{lemma}

\begin{proof}
Let $s(n) = o(n)$ be a complexity function, and consider the restriction of \msc{CSP} to instances with at most $s(n)$ tuples. We will show that this problem is solvable in time $\DOM^{s(n)}|I|^{O(1)}$.
Consider the algorithm $A$ that, for each tuple in a constraint, branches on whether or not the tuple is satisfied by the satisfying assignment sought. A branch in which more than one tuple in any constraint is selected as satisfied is rejected, and likewise for a branch in which no tuple in a constraint is selected. For each remaining branch, the algorithm checks if the assignment to the variables stipulated by the branch is consistent. If it is, the algorithm accepts; the algorithm rejects if no branch corresponds to a consistent assignment. Clearly, the algorithm $A$ is correct, and runs in time $2^{s(n)}|I|^{O(1)}=\DOM^{s(n)}|I|^{O(1)}$.
\end{proof}

Noting that the number of tuples is a lower bound for the instance size, the following theorem follow from Lemma~\ref{lem:hardnesssize} and Lemma~\ref{lem:tupleseasy}:

\begin{theorem}\label{thm:tuples}
\msc{CSP} is solvable in subexponential-time for instances in which the
number of tuples is $o(n)$, and unless the ETH fails, is not solvable in
subexponential-time if the number of tuples in the instances is
$\Omega(n)$.
\end{theorem}
Next, we show that the subexponential-time solvability of \msc{Boolean CSP} with linear size, or with linear number of tuples, is equivalent to the statement that the ETH fails. We first need the following lemma.

\begin{lemma}
\label{lem:serflineartobddarity}
If the ETH fails then  \msc{Boolean CSP} with linear number of tuples is solvable in subexponential time.
\end{lemma}

\begin{proof}
We give a serf-reduction from \msc{Boolean CSP} with linear number of tuples to \msc{Boolean $r$-CSP} for some constant $r \geq 3$ to be specified below. The statement will then follow from Proposition~\ref{thm:boundedarity}.

Let $s(n) \leq cn$ be a complexity function, where $c > 0$ is a
constant. Consider the restriction of \msc{Boolean CSP} to instances in
which the number of tuples is at most $cn$; we will refer to this
problem as \msc{Boolean Linear Tuple CSP}. Let $0 < \epsilon < 1$ be
given. Choose a positive integer-constant $d$ large enough so that the
unique root of the polynomial $x^d - x^{d-1} -1$ in the interval $(1,
\infty)$ is at most $2^{\epsilon/c}$. (The uniqueness of the root was
shown \cite[Lemma 4.1]{ChenKanjJia01}, and the fact that the root
converges to 1 as $d \longrightarrow \infty$ can be easily verified.)
Let $I$ be an instance of \msc{Boolean Linear Tuple CSP}. We will assume
that, for any constraint $C$ in $I$, and any two variables $x, y$ in
$C$, there must be at least one tuple in $C$ in which the values of $x$
and $y$ differ. If not, then the values of $x$ and $y$ in any assignment
that makes $I$ consistent have to be the same; in this case we remove all
tuples from $I$ in which the values of $x$ and $y$ differ, replace $y$ with $x$ in every constraint in $I$, and simplify $I$ accordingly (if a
constraint becomes empty during the above process then we reject $I$).

We now apply the following branching procedure to $I$. For each constraint $C$ in $I$ with more than $d$ tuples, pick a tuple $t$ in $C$ and branch on whether or not $t$ is satisfied in an assignment that makes $I$ consistent (if such an assignment exists). In the branch where $t$ is satisfied, remove $C$ from $I$, remove every tuple in $I$ in which the value of a variable that appears in $C$ does not conform to the value of the variable in $t$, and finally remove all variables in $C$ from $I$ and its tuples (if a constraint becomes empty reject $I$). In the branch where $t$ is not satisfied, remove $t$ from $C$. Note that each branch either removes a tuple or removes at least $d$ tuples. We repeat the above branching until each constraint in the resulting instance contains at most $d$ tuples. The above branching can be depicted by a binary search tree whose leaves correspond to all the possible outcomes from the above branching. The number of the leaves in the search tree is $O(x_{0}^{cn})$, where $x_0$ is the root of the polynomial $x^d - x^{d-1} -1$ in the interval $(1, \infty)$. (The branching vector is not worse than $(1, d)$.) By the choice of $d$, the number of leaves in the search tree is $O(2^{\epsilon n})$. Let $I'$ be the resulting instance at a leaf of the search tree. We claim that the $\ARITY$ of $I'$ is at most $2^d$. Suppose not, and let $C$ be a constraint in $I'$ whose $\ARITY$ is more than~$2^{d}$. Pick an arbitrary ordering of the tuples in~$C$, and list them as $t_1, \ldots, t_s$, where $s \leq d$. For each variable in $C$, we associate a binary sequence of length $s$ whose $i$th bit is the value of the variable in $t_i$. Since the $\ARITY$ is more than $2^d$, the number of binary sequences is more than $2^d$. Since the length of each sequence is $s \leq d$, by the pigeon-hole principal, there exist two binary sequences that are identical. This contradicts our assumption that no constraint has two variables whose values are identical in all the tuples of the constraint. It follows that the instance $I'$ is an instance of \msc{Boolean $2^d$-CSP}. Since the number of variables in $I'$ is at most that of $I$, and the number of leaves in the search tree is $O(2^{\epsilon n})$, we have a serf-reduction from \msc{Boolean Linear Tuple CSP} to \msc{Boolean $r$-CSP} for some constant~$r$.
\end{proof}

Lemma~\ref{lem:hardnesssize}, combined with
Lemma~\ref{lem:serflineartobddarity} after noting that the size is an
upper bound on the number of tuples, give the following result.

\begin{theorem}
\label{thm:tupleeth}
\msc{Boolean CSP} with linear number of tuples is solvable in subexponential time if and only if the ETH fails.
\end{theorem}

\begin{theorem}
\label{thm:linearcsp}
The \msc{Boolean CSP} with linear size is solvable in subexponential time if and only if the ETH fails.
\end{theorem}

\section{Treewidth and number of constraints}
\label{sec:constraints}

In this section we characterize the subexponential-time complexity of
\msc{CSP} with respect to the \emph{treewidth} of
certain graphs that model the interaction of variables and constraints.
Many NP-hard problems on graphs
become polynomial-time solvable for graphs whose treewidth is bounded
by a constant. For a definition of treewidth we refer to other sources
\cite{Bodlaender98}.
\citex{Freuder90} showed that \msc{CSP} is polynomial-time
solvable if a certain graph associated with the instance, the
\emph{primal graph}, is of bounded treewidth. The primal graph
associated with a \msc{CSP} instance $I$ has the variables in $I$ as its
vertices; two variables are joined by an edge if and only if they occur
together in the scope of a constraint.  Freuder's result was
generalized in various ways, and other restrictions on the graph
structure of CSP instances have been considered
\cite{GottlobLeoneScarcello00,Marx10a}. If the treewidth of the primal
graph is bounded, then so is the arity of the constraints.
The \emph{incidence graph} provides a more general graph model, as
it includes instances of unbounded arity even if the treewidth is
bounded.  The incidence graph associated with $I$ is a bipartite graph
with one partition being the set of variables in~$I$ and the other
partition being the set of constraints in~$I$; a variable and a
constraint are joined by an edge if and only if the variable occurs in
the scope of the constraint. For a \msc{CSP} instance, we denote by
$\TW$ the treewidth of its primal graph and by $\TW^*$ the treewidth of
its incidence graph.

As shown by \citex{Bodlaender96}, there exists for every fixed~$k$
a linear time algorithm that checks if a graph has
treewidth at most~$k$ and, if so, outputs a tree decomposition of
minimum width. It follows that we can check whether the treewidth of a graph is $O(1)$ in polynomial time.

\begin{lemma}\label{lem:tw*}
\msc{CSP} is solvable in polynomial time for instances whose incidence treewidth $\TW^*$ is $O(1)$.
\end{lemma}
\begin{proof}
If the $\TW^*$ is $O(1)$ then the {\em hypertree-width} is also $O(1)$
\cite{GottlobLeoneScarcello00}, and \msc{CSP} is solvable in
polynomial-time if the hypertree-width is
$O(1)$~\cite{GottlobLeoneScarcello02b}. Combining the preceding
statements gives the lemma.
\end{proof}

\aiversion{\enlargethispage*{5mm}}

\begin{lemma}
\label{lem:hardnesstw*}
Unless the ETH fails,  \msc{CSP} is not solvable in subexponential-time if the number of constraints is $\omega(1)$.
\end{lemma}

\begin{proof}
Let $\lambda(n) = \omega(1)$ be a complexity function. We show that, unless the ETH fails, the restriction of \msc{CSP} to instances in which $\CONS \leq \lambda(n)$, denoted \msc{CSP$_{\lambda}$} is not solvable in $\DOM^{o(n)}$ time. By Proposition~\ref{thm:boundedarity}, it suffices to provide a serf-reduction from \msc{Boolean 3-CSP} with a linear number of constraints to \msc{Boolean CSP$_{\lambda}$}.

Let $I$ be an instance of \msc{Boolean CSP} in which $\CONS =n' \leq cn$, where $c > 0$ is a constant. Let $C_1, \ldots, C_{n'}$ be the constraints in $I$; we partition these constraints arbitrarily into $\lfloor \lambda(n) \rfloor$ many groups ${\cal C}_1, \ldots, {\cal C}_r$, where $r \leq \lfloor \lambda(n) \rfloor$, each containing at most $\lceil n'/\lambda(n) \rceil$ constraints. The serf-reduction $A$ works as follows.
$A$ ``merges'' all the constraints in each group ${\cal C}_i$, $i=1, \ldots, r$, into one constraint $C'_i$ as follows. The variable-set of $C'_i$ consists of the union of the variable-sets of the constraints in ${\cal C}_i$. For each constraint $C$ in ${\cal C}_i$, iterate over all tuples in $C$. After selecting a tuple from each constraint in ${\cal C}_i$, check if all the selected tuples are consistent, and if so merge all these tuples into a single tuple and add it to $C'_i$. By merging the tuples we mean form a single tuple over the variables in these tuples, and in which the value of each variable is its value in the selected tuples (note that the values are consistent).
Since each constraint in $I$ has arity at most 3, and hence contains at most 8 tuples, and since each group contains at most $\lceil n'/\lambda(n) \rceil$ constraints, $C'_i$ can be constructed in time
$8^{\lceil n'/\lambda(n) \rceil} n'^{O(1)}=2^{o(n)}$, and hence, all the constraints $C'_1, \ldots, C'_r$ can be constructed in time $2^{o(n)} n^{O(1)}=2^{o(n)}$. We now form the instance $I'$ whose variable-set is that of $I$, and whose constraints are $C'_1, \ldots, C'_r$. Since  $r \leq \lfloor \lambda(n) \rfloor$, $I'$ is an instance of \msc{CSP$_{\lambda}$}. Moreover, it is easy to see that $I$ is consistent if and only if $I'$ is. Since $I'$ can be constructed from $I$ in subexponential time and the number of variables in $I'$ is at most that of $I$, it follows that $A$ is a serf-reduction from \msc{Boolean 3-CSP} with a linear number of constraints to \msc{CSP$_{\lambda}$}.
\end{proof}

Since $\TW^* = O(\CONS)$, (removing the vertices corresponding to the constrains from the incidence graph results in an independent set)
Lemma~\ref{lem:tw*} and Lemma~\ref{lem:hardnesstw*} give the following result.

\begin{theorem}
\label{thm:cons}
\msc{CSP} is solvable in polynomial time for instances with $O(1)$ constraints, and unless the ETH fails, is not solvable in subexponential-time if the number of constraints is $\omega(1)$.
\end{theorem}

\begin{theorem}
\label{thm:incidencetw}
\msc{CSP} is solvable in polynomial time for instances whose incidence treewidth $\TW^*$ is $O(1)$, and unless the ETH fails, is not solvable in subexponential-time for instances whose $\TW^*$ is $\omega(1)$.
\end{theorem}

\begin{theorem}
\label{thm:primaltw}
\msc{CSP} is solvable in subexponential-time for instances whose primal
treewidth $\TW$ is $o(n)$, and is not solvable in subexponential-time
for instances whose $\TW$ is $\Omega(n)$ unless (the general) \msc{CSP} is
solvable in subexponential time.
\end{theorem}
\begin{proof}
The fact that \msc{CSP} is solvable in subexponential time if $\TW = o(n)$ follows from the facts that: (1) we can compute a tree decomposition of width at most $4 \cdot \TW$ in time $2^{4.38 \TW}|I|^{O(1)}$~\cite{Amir10}, and (2) \msc{CSP} is solvable in time $O(\DOM^{\TW}) |I|^{O(n)}$~\cite{Freuder90}.

Let $s(n) = cn$, where $c > 0$ is a constant, and consider the restriction of \msc{CSP} to instances whose $\TW$ is at most $s(n)$, denoted \msc{Linear-$\TW$-CSP}. Note that the number of vertices in the primal graph is $n$, and hence $\TW \leq n$. Therefore, if $c \geq 1$, then the statement trivially follows. Suppose now that $c < 1$, and let $I$ be an instance of \msc{CSP} with $n$ variables. By ``padding'' $\lceil 1/c \rceil$ disjoint copies of $I$ we obtain an instance $I'$ that is equivalent to $I$, whose number of variables is $N'=\lceil 1/c \rceil n$, and whose $\TW$ is the same as that of $I$. Since the $\TW$ of $I$ is at most $n$, it follows that the $\TW$ of $I'$ is at most $cN'$, and hence $I'$ is an instance of \msc{Linear-$\TW$-CSP}. This gives a serf-reduction from \msc{CSP} to \msc{Linear-$\TW$-CSP}.
\end{proof}

We note that the hypothesis ``\msc{CSP} is solvable in
subexponential time'' in the above theorem implies that
``ETH fails'' by
Proposition~\ref{thm:boundedarity}, and implies that
\msc{CNF-Sat} has a nonuniform subexponential-time algorithm by
Theorem~\ref{thm:sattocsp}. We also note that the difference between the subexponential-time
complexity of \msc{CSP} with respect to the two structural parameters
$\TW$ and $\TW^*$: Whereas the threshold function for the
subexponential-time solvability of \msc{CSP} with respect to $\TW$ is
$o(n)$, the threshold function with respect to $\TW^*$ is $O(1)$.

\aiversion{\enlargethispage*{5mm}}
\section{Degree and arity}
\label{sec:degarity}
In this section we give characterizations of the subexponential-time complexity of \msc{CSP} with respect to the degree and the arity. The proofs are omitted.

\begin{theorem}
\label{thm:degarity}
Unless ETH fails,  \msc{CSP}  is not solvable in subexponential-time if $\DEG \geq 2$.
\end{theorem}
\longversion{
\begin{proof}
The statement follows from the proof of Theorem~\ref{thm:boundedarity} after noting that, by Lemma~\ref{lem:subexpnm}, one can use \msc{$r$-CNF-Sat} with degree at most $3$ (after introducing a linear number of new variables) in the reduction. This will result in instances of \msc{Boolean $r$-CSP} with degree at most $3$ as well. Now for each variable $x$ of degree $3$ in an instance of \msc{Boolean $r$-CSP}, we introduce two new variables $x', x''$, and add a constraint whose variables are $\{x, x', x''\}$, and containing the two tuples $(0, 0, 0)$ and $(1, 1, 1)$; this constraint stipulates that the values of $x, x', x''$ be the same. We then substitute the variable $x$ in one of the constraints it appears in with $x'$, and in another constraint that it appears in with $x''$. Therefore, in the new instance, the degree of each of $x, x', x''$ becomes 2. After repeating this step to every variable of degree 3, we obtain an instance of \msc{Boolean $r$-CSP} in which the degree of each variable is at most 2. Since the increase in the number of variables is linear, a subexponential-time algorithm for \msc{Boolean $r$-CSP} with degree at most 2 implies a subexponential-time algorithm for \msc{$r$-CNF-Sat}.
\end{proof}
}
\longversion{As mentioned in Section~\ref{sec:intro},}
There is a folklore reduction
from an instance of \msc{3-Colorability} with $n$ vertices that results
in an instance of \msc{CSP} with $n$ variables, $\ARITY =2$, and $\DOM
=3$. Since the \msc{3-Colorability} problem is SNP-complete under
serf-reductions~\cite{ImpagliazzoPaturiZane01}, we get:

\begin{theorem}
\label{thm:degarity}
Unless ETH fails,  \msc{CSP}  is not solvable in subexponential-time if $\ARITY \geq 2$ (and $\DOM \geq 3$).
\end{theorem}
\longversion{
\begin{proof}
We will show that a subexponential-time algorithm for \msc{CSP} with $\ARITY =2$ and $\DOM=3$ implies that the \msc{3-Colorability} problem is solvable in subexponential time. Since the \msc{3-Colorability} problem is SNP-complete under serf-reductions~\cite{ImpagliazzoPaturiZane01}, the statement of the theorem will follow. Recall that the \msc{3-Colorability} problem asks if the vertices of a given graph can be colored with at most 3 colors so that no two adjacent vertices are assigned the same color.

The reduction is folklore. Given an instance of $G=(V, E)$ of
\msc{3-Colorability}, we construct an instance~$I$ of \msc{CSP} as
follows. The variables of $I$ correspond to the vertices of $G$, and the
domain of $I$ corresponds to the color-set $\{1, 2, 3\}$. For every edge
of the graph we construct a constraint of $\ARITY =2$ over the two
variables corresponding to the endpoint of the edge. The constraint
contains all tuples corresponding to valid colorings of the endpoints of
the edge. It is easy to see that $G$ has a 3-coloring if and only if $I$
is consistent. Since for the instance $I$ we have $\VARS= n$, $\ARITY
=2$, and $\DOM =3$, an algorithm running in time $\DOM^{o(n)}$ for
\msc{CSP} with $\ARITY =2$ and $\DOM =3$ would imply a
subexponential-time algorithm for \msc{3-Colorability}.
\end{proof}
We note that \msc{CSP} with $\DOM=2$ and $\ARITY = 2$ is solvable in polynomial time via a simple reduction to \msc{2-CNF-Sat}.
}

\section{Conclusion}

We have provided a first analysis of the subexponential-time complexity
of CSP under various restrictions.  We have obtained several tight
thresholds that dictate the subexponential-time complexity of CSP. These
tight results are summarized in the following table.
\aiversion{\vspace*{-2mm}}
\begin{center}\small
\begin{tabular}{@{}l@{~~~~~~~~~~~}ll@{}}
  \toprule
  CSP $\in$ SUBEXP & CSP $\notin$ SUBEXP &
  Result\\
                   &  (assuming the ETH)   & \\
\midrule
$\TUPLES\in o(n)$ &  $\TUPLES\in \Omega(n)$ & Theorem~\ref{thm:tuples}\\
$\CONS\in O(1)$ (even in P) &  $\CONS\in \omega(1)$ & Theorem~\ref{thm:cons}\\
$\TW^*\in O(1)$ (even in P) &  $\TW^*\in \omega(1)$ & Theorem~\ref{thm:incidencetw}\\
$\TW\in o(n)$ &  $\TW\in \Omega(n)$ & Theorem~\ref{thm:primaltw}\\
\bottomrule
\end{tabular}
\end{center}
\medskip
Furthermore, we have linked the subexponential-time complexity of
\msc{CSP} with bounded arity to \msc{Clique}, and \msc{CSP} with
bounded domain size to \msc{CNF-Sat}. These results suggest that these
restrictions of \msc{CSP} may be ``harder than'' \msc{$k$-CNF-Sat}---for
which a subexponential-time algorithm would lead to the failure of
the ETH---with respect to subexponential-time complexity. It would be
interesting to provide stronger theoretical evidence for
this separation.

\longversion{
%

}
\aiversion{
\cleardoublepage
\pagebreak

}

\end{document}